\documentclass[12pt]{article}
\usepackage{amsthm,amssymb,amsfonts,epsfig}
\usepackage{float}
\restylefloat{table}
\usepackage{amsmath}
\usepackage{amssymb}
\usepackage{graphicx}
\graphicspath{{img/}}
\usepackage{subfigure}
\usepackage{algorithm}
\usepackage{algpseudocode}
\usepackage{setspace}
\usepackage{hyperref}
\usepackage{color}

\begin{document}
\newcommand{\elimina}[1]{}
\newcommand{\piB}{{\tt B}}
\newcommand{\piO}{{\mbox{\sc opt}}}
\newcommand{\costB}{Cost_{\tt B}}
\newcommand{\costO}{Cost_{\mbox{\sc opt}}}
\newtheorem{definition}{Definition}
\newtheorem{theorem}[definition]{Theorem}
\newtheorem{lemma}[definition]{Lemma}
\newtheorem{corollary}[definition]{Corollary}
\newtheorem{example}[definition]{Example}

\date{}
\title{A Distributed Message-Optimal Assignment on Rings\footnote{A preliminary version of this paper has been
    presented at IPDPS06.} }
\author{Gianluca De Marco\thanks{Dipartimento di Informatica, Universit\`a di  Salerno, Italy. e-mail:
{\tt  demarco@dia.unisa.it}} \and  
Mauro Leoncini\thanks{Dipartimento di Scienze Fisiche, Informatiche e Matematiche, Universit\`a di 
Modena e Reggio Emilia, Italy, and Istituto di Informatica e Telematiche, CNR
Pisa, Italy. e-mail: {\tt leoncini@unimore.it}} \and Manuela
Montangero\thanks{Dipartimento di Scienze Fisiche, Informatiche e Matematiche, Universit\`a di 
Modena e Reggio Emilia, Italy, and Istituto di Informatica e Telematiche, CNR
Pisa, Italy. e-mail: {\tt manuela.montangero@unimore.it}}}
\maketitle

\begin{abstract}

  Consider a set of items and a set of $m$ colors, where each item is
  associated to one color. Consider also $n$ computational agents
  connected by a ring.  Each agent holds a subset of the items and
  items of the same color can be held by different agents.  We analyze
  the problem of distributively assigning colors to agents in such a
  way that (a) each color is assigned to one agent only and (b) the
  number of different colors assigned to each agent is minimum. Since
  any color assignment requires the items be distributed according to
  it ({\em e.g.} all items of the same color are to be held by only
  one agent), we define the cost of a color assignment as the amount
  of items that need to be moved, given an initial allocation.  We
  first show that any distributed algorithm for this problem requires
  a message complexity of $\Omega(n\cdot m)$ and then we exhibit an
  optimal message complexity algorithm for synchronous rings that in
  polynomial time determines a color assignment with cost at most
  three times the optimal. We also discuss solutions for the
  asynchronous setting. Finally, we show how to get a better cost
  solution at the expenses of either the message or the time
  complexity.  \vspace*{1cm}

\noindent
{\bf keywords:} algorithms; distributed computing; leader election; ring.
\vspace*{2cm}
\end{abstract}
\newpage
\section{Introduction}

We consider the following problem.  We are given a set of
computational agents connected by a (physical or logical)
ring\footnote{An importand example of logical architecture is given by
  the set of ring shaped nodes of a Distributed Hash Table.}, and a
set of items, each associated to one color from a given set.
Initially each agent holds a set of items and items with the same
color may be held by different agents ({\em e.g.} see
Fig~\ref{fig:problem}.(a)).  We wish the agents to agree on an
assignment of colors to agents in such a way that each color is
assigned to one agent only and that the maximum over all agents of the
number of different colors assigned to the same agent is minimum. We
call this a {\em balanced assignment}: Fig~\ref{fig:problem}.(b) and
Fig~\ref{fig:problem}.(c) show two possible balanced
assignments. Among all such assignments, we seek the one that
minimizes the total number of items that agents have to collect from
other agents in order to satisfy the constraints. For example, agent
$a_0$ in Fig~\ref{fig:problem}.(b) is assigned colors $\nabla$ and
$\spadesuit$, and therefore needs just to collect four items colored
$\nabla$, since no other agent has items colored $\spadesuit$.

\medskip

\begin{figure}[htbp]
\label{fig:problem}
\centerline{\epsfig{file=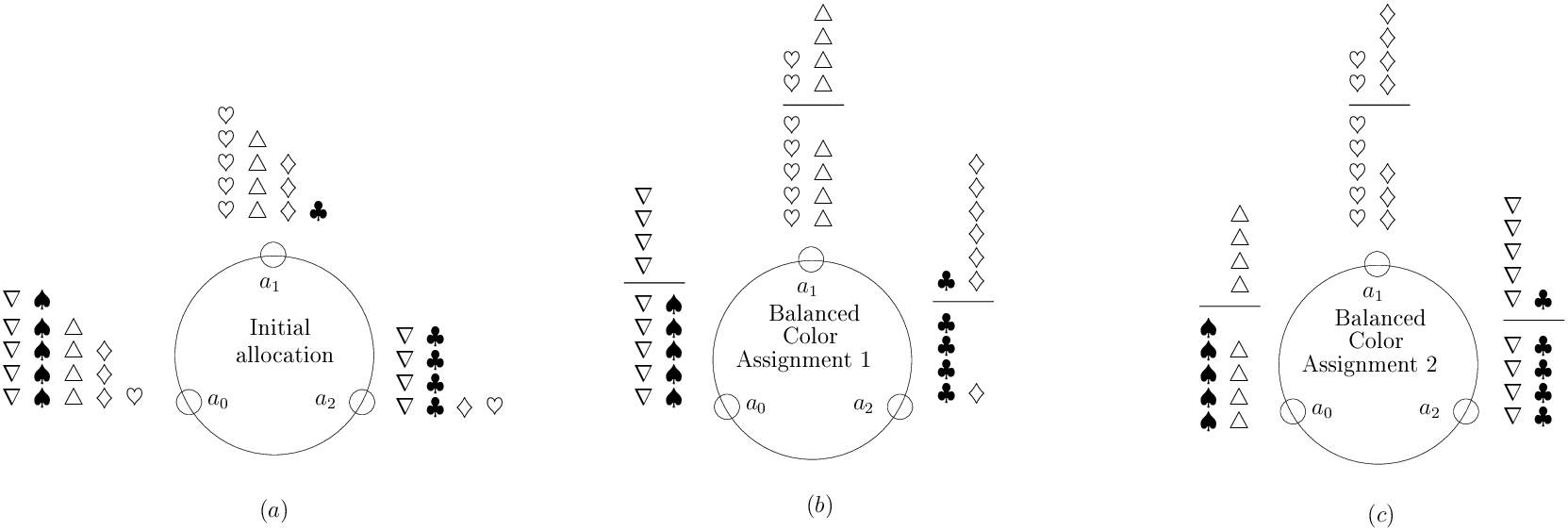, width=6in}}
\caption{\small Three agents: $a_0$, $a_1$, $a_2$, and
six colors: $\nabla, \diamondsuit,\heartsuit,\triangle,\spadesuit,\clubsuit$. $(a)$
is the initial allocation, while $(b)$ and $(c)$ are two possible balanced color
assignments. Items above the line are those that the agent collects from the
others. Therefore their total number is the cost of the assignment.
The assignment in $(b)$ costs $(4 \times\nabla )+(2\times\heartsuit +
4\times\triangle ) + (1\times\clubsuit + 6\times\diamondsuit) = 17 $ items,
while the assignment in $(c)$ costs
 $(4 \times\triangle )+(2\times\heartsuit +
4\times\diamondsuit ) + (5\times\nabla + 1\times\clubsuit) = 16 $ items.
}
\end{figure}

The problem can be formalized as follows.  Let ${\cal A}=\{a_0, \ldots
a_{n-1}\}$ be a set of $n$ agents connected by a ring and
let ${\cal C}=\{ c_0,\ldots ,c_{m-1} \}$ be a set of $m$ colors.  Let
$Q_{j,i}\geq 0$ be the number of items with color $c_j$ initially held by
agent $a_i$, for every $j=0,\ldots, m-1$, and for every $i= 0,\ldots,
n-1$. 

\begin{definition}[Balanced Coloring]
\label{def:BC}

A Balanced Coloring is an assignment $\pi:\{ 0,\ldots,m-1\} \rightarrow 
\{ 0,\ldots, n-1\}$ of the $m$ colors to the $n$ agents in such a way that:

\begin{itemize}

\item for every color $c_j$, there is at least one
agent $a_i$ such that $\pi (j) = i$;

\item for every agent $a_i$, 
$\lfloor \frac{m}{n} \rfloor \leq  |\{c_j \mid  \pi (j) = i  \}| \leq \lceil
\frac{m}{n} \rceil$; {\em i.e.}, the number of color assigned to
agents has to be balanced. In particular, $\lfloor \frac m n
\rfloor$ colors are assigned to $[ \left( \lfloor \frac m n \rfloor
+ 1 \right) n - m ] $
agents, and $\lfloor \frac m n \rfloor + 1$ colors
to the remaining $\left( m - \lfloor \frac m n \rfloor
  n \right) $
agents. 
\end{itemize}

\end{definition}

Any Balanced Coloring then assigns almost the same number of
colors to each agent, and when $m$ is a multiple of $n$, then each agent
is assigned exactly the same number of colors.

\begin{definition}[Distributed Balanced Color Assignment Problem]
\label{def:cost}
The Distributed Balanced Color Assignment Problem aims at distributively
finding a Balanced Coloring of minimum cost, where   
the {\em cost} of a Balanced Coloring $\pi:\{ 0,\ldots,m-1\} \rightarrow 
\{ 0,\ldots, n-1\}$ is defined as

\begin{equation} \label{cost}
Cost(\pi) = \sum_{j=0}^{m-1} \sum_{i=0, \atop i\not= \pi(j)}^{n-1}
Q_{j,i} .
\end{equation}

\end{definition}

The cost of the
optimal assignment will be denoted by $\costO$. The {\em approximation ratio}
of a sub-optimal algorithm {\tt A} is the quantity
$\frac{Cost_{\mbox{\tt A}}}{\costO}$, where  
$Cost_{\mbox{\tt A}}$  is the cost of the solution computed by {\tt A}.

\paragraph{Motivations.} The scenario defined above may arise in many practical situations in
which a set of agents independently search a common space 
(distributed crawlers, sensor networks, peer-to-peer agents, etc)  and then have to
reorganize the retrieved data (items) according to a given
classification (colors), see for example \cite{KCRMC, PKKMH, SMKKB}.  In these cases, determining a
distributed balanced color assignment may guarantee specialization by
category with maximal use of data stored in local memory or
balanced computational load of agents minimizing the communication
among agents. 

A similar scenario may also
arise in computational economics~\cite{CCM}.  The distributed balanced color
assignment formalizes a combinatorial auction problem where agents are
the bidders and colors represent auction objects.  The number of items
that an agent holds for each color can be interpreted as a measure of
{\em desire} for certain objects (colors).  Balancing the number of
colors per agent and minimizing the cost guarantees the maximum
bidders satisfaction.

\paragraph{The model.} We assume that the agents in ${\cal A}=\{a_0,
\ldots, a_{n-1}\}$ are connected by a ring: agent $a_i$ can
communicate only with its two neighbors $a_{(i+1) \mod \, n}$
(clockwise) and $a_{(i-1) \mod \, n}$ (anti-clockwise). We assume that
each agent knows $n$ (the number of agents), and ${\cal C}$ (the set
of colors).  Each agent $a_i$ is able to compute $p_i = \max_{0\leq
  j\leq m-1} Q_{j,i}$ independently, {\em i.e.}, the maximum number of
items it stores having the same color, while $p = max_{0\leq i \leq
  n-1} p_i$ is unknown to the agents.

We will consider both synchronous and asynchronous rings,
always specifying which  case we are working with or if results hold
for both models. 

For synchronous and asynchronous rings, we measure message complexity
in the standard way (cf. \cite{NL,P}), {\em i.e.}, we assume that
messages of bit length at most $c\log n$, for some constant $c$
(called {\em basic messages}), can be transmitted at unit cost.  One
message can carry at most a constant number of agent ID's. Non basic
messages of length $L$ are allowed, and we charge a cost $c'
\left\lceil L/\log n \right\rceil$ for their transmission, for some
constant $c'$.

For what concerns time complexity, in the synchronous case we assume
that agents have access to a global clock and that the distributed
computation proceeds in rounds.  In each round any agent can check if
it has received a message (sent in the previous round), make some
local computation, and finally send a message.  In the asynchronous
case agents don't have access to a global clock, but the distributed
computation is {\it event driven} (``upon receiving message $\alpha$,
take action $\beta$'').  A message sent from one agent to another will
arrive in a finite but {\it unbounded} amount of time.

Throughout the paper we will use the generic term {\it time unit} to
designate the time needed for a message to traverse a link both in the
synchronous and asynchronous case: for the synchronous case a time
unit (also called round or time slot) is the time elapsed between two
consecutive ticks of the clock; for the asynchronous setting a time
unit can be any bounded finite amount of time.  Nevertheless, in both
cases the time complexity can be simply measured as the number of time
units needed to complete the algorithm's execution.

\paragraph{Outline of the results.} The goal of this paper is to
analyze the efficiency with which we can solve the Distributed
Balanced Color Assignment problem. In Section~\ref{pre} we discuss
some related problems and show the equivalence with the so called
weighted $\beta$-assignment problem in a centralized
setting~\cite{CL}. We also show that a brute force approach that first
gathers all information at one agent, then computes the solution
locally and finally broadcasts it, has a high message complexity of
$O(m n^2 \log p / \log n)$.  Fortunately, we can do better than this.
In Section~\ref{LowerBound} we give an $\Omega(mn)$ lower bound on the
message complexity to determine a feasible solution (suitable for both
synchronous and asynchronous cases).  In Section~\ref{Alg} we present
an algorithm that finds a feasible solution to the problem
whose message complexity is $O(mn\log m / \log n)$, which is then optimal
when $m$ is bounded by a polynomial in $n$. 

Interestingly enough, message complexity is never affected by the
value $p$, while running time is.  
We then show how to adapt the
algorithm to work also in the asynchronous case at the expenses of a
slight increase in message complexity; this time the messege cost
depends also on $p$, but the asymptotic bound is affected only when 
$p$ is very large (i.e., only if $p \not\in O(m^m)$). In Section~\ref{Approx} we
show that the proposed algorithm (both synchronous and asynchronous
versions) computes a { Balanced Coloring} whose cost is only a factor of
three off the optimal one, and we also show that the analysis of the
approximation is tight. Finally, we show that we can find { Balanced Colorings} 
with a better approximation ratio at the expenses of the
message and/or time complexity.

A preliminary version of this work appeared in \cite{DLM}. 
In the previous version it was assumed that parameter $p$ (the maximum
number of items of a given color) was known to the computational agents. 
Since in practical situations it is difficult to have a good estimate
of such a global parameter, in this new version we removed this assumption.  
This required both new algorithmic ideas and technical efforts. 
The algorithm for the asynchronous communication model was also not contained in the preliminary version. 
Finally, we enriched the proof of the lower bound with new 
insights that could be useful for further generalizations to different network topologies.

\section{Related problems and centralized version}
\label{pre}

In this section we relate the Distributed Balanced Color Assignment
problem to known matching problems that have been well studied in
centralized settings. We will first show that when $m=n$ our problem
is equivalent to a maximum weight perfect matching problem on complete
bipartite graphs. On the other hand, when $m\geq n$, our problem
reduces to the weighted $\beta$-assignment problem.

The class of $\beta$-assignment problems has been introduced by Chang
and Lee \cite{CL}, in the context of the problems of assigning jobs to
workers, in order to incorporate the problem of balancing the work
load that is given to each worker. In the weighted $\beta$-assignment
problem one aims at minimizing the maximum number of jobs assigned to
each worker.

The interested reader can find useful references on these problems,
their complexity, and related approximation issues in
\cite{AAB,CH,KR,SM,SR}.

\medskip

We associate to agents and colors the complete bipartite graph on $n
+ m$ vertices, which we denote by $G=({\cal C}, {\cal A}, {\cal
C}\times{\cal A})$. We add weights to $G$ as follows: the weight of
the edge joining agent $a_i$ and color $c_j$ is $Q_{j,i}$.

\paragraph{ Case $m = n$.}
Given a graph $(V,E)$, a perfect matching is a subset $M$ of edges in
$E$ such that no two edges in $M$ share a common vertex and each
vertex of $V$ is incident to some edge in $M$. When edges of the graph
have an associated weight, then a maximum weight perfect matching is a
perfect matching such that the sum of the weights of the edges in the
matching is maximum.

\begin{lemma}  \label{maxmatching}
{ When $m=n$}, a maximum weight perfect matching on $G$ is a minimum cost
solution to the balanced color assignment problem.
\end{lemma}

\begin{proof}
  Given a perfect matching $E\subseteq {\cal E}={\cal C}\times{\cal
    A}$ on $G$, for every $(c_j,a_i)\in E$ we assign color $c_j$ to
  agent $a_i$. As $G$ is complete and $E$ is a perfect matching on
  $G$, every color is assigned to one and only one agent and
  vice-versa. Moreover, the cost of any color assignment $E$ can be
  written as $\sum_{e\in {\cal E}\setminus E} w(e)$, and this
  expression achieves its minimum when $E$ is a maximum weight perfect
  matching.
\end{proof}

\medskip

Finding matchings in graphs is one of the most deeply investigated
problems in Computer Science and Operations Research (see
\cite{Lovasz} for a comprehensive description of the different
variants, theoretical properties, and corresponding algorithms). The
best algorithm known to find a perfect matching in a bipartite graph
is due to Hopcroft and Karp \cite{HK}, and runs in
$O\left(|E|\sqrt{|V|}\right)$ time, where $V$ and $E$ denote the
vertex and edge sets,
respectively. 
The best known algorithm for finding a maximum weight
perfect matching is the {\em Hungarian method}, due to Kuhn \cite{K},
which runs in time $O(n^3)$.

\paragraph{Case $m\geq n$.}
The $\beta$-assignment problem is defined on a bipartite graph
$G=(S,T,E)$ where $(S,T)$ is the bipartition of the vertex set. A
$\beta$-assignment of $S$ in $G$ is a subset of the edges $X
\subseteq E$ such that, in the induced subgraph $G'=(S,T,X)$, the degree
of every vertex in $S$ is exactly one. Let $\beta(X)$ be the
maximum degree, in $G'$, of vertices in $T$ and let $\beta(G)$ be
the minimum value of $\beta(X)$ among all possible
$\beta$-assignments $X$. The weighted $\beta$-assignment problem
consists of finding a $\beta$-assignment $X$ with $\beta(X) =
\beta(G)$ which maximizes the total weight of the edges in $X$.
The following lemma is straightforward.

\begin{lemma}
The balanced color assignment problem is a weighted
$\beta$-assignment of ${\cal C}$ in the complete bipartite graph
$G=({\cal C}, {\cal A}, {\cal C} \times {\cal A})$, with $\beta(G)
= \lceil m/n \rceil$.
\end{lemma}

The fastest known algorithm to solve the weighted $\beta$-assignment
problem is due to Chang and Ho \cite{CH} and runs in $O(\max\lbrace
|S|^2|T|,|S||T|^2\rbrace)$ time, which in our case gives the bound
$O(m^2n)$.

While the maximum weighted perfect matching problem (and its variants) has
been widely investigated in the distributed setting (see
\cite{HKP1,HKP2}), no distributed results are known for the weighted
$\beta$-assignment problem.

\paragraph{A brute force approach.}
A brute force distributed solution to the problem can be obtained by
asking all the agents to send their color information to one specific
agent (a priori chosen or elected as the leader of the ring); such an
agent will then solve the problem locally and send the solution
back to all the other agents.  The factor dominating the message
complexity of the algorithm above is the information collecting stage.
Indeed, each agent sends $O(m)$ non-basic messages, each corresponding
to $O(\log p / \log n)$ basic messages, through $O(n)$ links, on the
average. This results in a message complexity of $O(m n^2 \log p /
\log n)$.  On the other hand, we might think of an algorithm in which
each agent selects the correct number of colors basing its choice just
on local information ({\em e.g.} its label). This requires no
communication at all, but, even if we are able to prove that the
agents agree correctly on a balanced coloring, we have no guarantee on
how good the solution is.  As we already said, we show that we can do
better than this.



\elimina{Can we do better?
In the next section we show a $\Omega(mn)$ lower bound on the message
complexity of the problem, and in Section~\ref{Alg} we describe, for
the synchronous ring, a $O(n\log p)$ time distributed algorithm with
cost at most three times the optimal showing optimal message
complexity under the reasonable hypothesis that $m \in O(n^c)$ for some constant $c$. 
We than show how to adapt the algorithm to the asynchronous case,
where an extra  }

\section{Lower bound on message complexity}
\label{LowerBound}

\elimina{In this section we prove a lower bound to the problem that
  applies to rings.  In particular, we will prove that agents need to
  exchange $\Omega (mD)$ messages to solve the problem, where $D$ is
  the diameter of the graph.

\medskip
}

In this section we prove a lower bound on the message complexity of
the problem that applies to both synchronous and asynchronous rings.
\elimina{ Before formally proving the lower bound we make one
  observations to clarify the nature of the problem.

\begin{obs}
Any solution in which the agents select their colors
according to the initial distribution of
items requires some communication among agents.
\end{obs}

Assume agents can agree on a balanced coloring without
communicating, {\em i.e.}, each agent is able to select
the colors that are assigned to it autonomously without
making errors. In such a situation, the choice
of agent $a_i$ must be independent of the particular values
of the $Q_{j,i}$'s, for otherwise, the agent cannot be
sure that no other agent chooses the same colors (suppose
the two have exactly the same initial distribution
of items). On the other hand, an assignment that is
done without taking the initial distribution of items
to agents into consideration cannot give any guarantee
on the cost. Consider, for example, the case of an
initial distribution such that each agent has items of
all colors except the ones that it selects. This leads
to a cost that is equal to the total amount of items.
Any assignment in which at least one agent is assigned
one of the colors it initially holds has a smaller cost.
We conclude that any solution in which the agents' selection
is done according to the initial distribution of
items requires some communication among agents. 

\medskip
} 

We prove the lower bound on a particular subset ${\cal I}$ of the
instances of the problem. Let $n$ be even and let $m=(nt)/2$, for some
integer $t$. Since we are only interested in asymptotic bounds, for the sake of simplicity, 
we will also assume that $m$ is a multiple of $n$, i.e. $t/2 = m/n$ is an integer.

For any agent $a_i$, let $a_{i'}$ denote the agent at
maximum distance from $a_i$ on the ring. In the following we say that
{\em a color is assigned to the pair} $(a_i,a_{i'})$, for $i=0,\ldots ,
n/2-1$, to mean that it is assigned to both agents of the pair.  We
also say that a set $\mathcal{C}$ of colors is assigned to agent $a$
iff all the colors in $\mathcal{C}$ are assigned to $a$.

Let $\{ {\cal C}_0, \ldots , {\cal C}_{n/2-1}\}$ be a partition of the
set of colors such that $|{\cal C}_j| = t$ for all $j= 0, \ldots,
n/2-1$.  Set ${\cal I}$ consists of all instances of the Distributed
Balanced Color Assignment Problem such that for any $i=0,\ldots ,
n/2-1$, the following two conditions hold:

(a) for any color $j \in {\cal C}_i = \{ i_1,\ldots , i_t\}$ 
both agents of pair $(a_i,a_{i'})$ hold at least one item of color
$j$, {\em i.e.} $Q_{j,i}>0, Q_{j,i'} > 0$;

(b) neither $a_i$ nor $a_{i'}$ hold colors not in $\mathcal{C}_i$.

\begin{lemma}\label{le:pair}
  Given an instance in ${\cal I}$, any optimal solution assigns to
  $(a_i,a_{i'})$ {\it only} colors from set ${\cal C}_i$, for
  $i=0,\ldots , n/2-1$.
\end{lemma}

\begin{proof}
  Consider any solution to an instance from set ${\cal I}$ that
  assigns to the agent $a_i$ a color $h_0$ initially held by some pair
  $(a_k,a_{k'})$, with $k\neq i$. Since any optimal solution is
  perfectly balanced on input instances of $\mathcal{I}$, there must
  be at least one color $h_1$ initially stored in $(a_i,a_{i'})$ that
  is assigned to some other agent, say $a_p$. The same argument can in
  turn be applied to $a_p$ and so on until (since the number of
  colors/agents is finite) we fall back on $a_k$.  Formally, there
  exists $0 \leq k \leq n/2-1$, $k \not = i$, such that $h_0 \in {\cal
    C}_k \not = {\cal C}_i$, and a sequence of indices $k_0,
  k_1,\ldots, k_l$, with $k_0 = k$, $k_1 = i$ and $k_{l+1} = k$, such
  that
\begin{itemize}
  \item color $h_0\in {\cal C}_{k_0} (= {\cal C}_k)$ is assigned to agent $a_{k_1} (= a_i)$;
  \item color $h_1\in {\cal C}_{k_1} (= {\cal C}_i)$ is assigned to agent $a_{k_2}$;\\
  \vdots
  \item color $h_l\in {\cal C}_{k_{l}}$  is assigned to agent $a_{k_{l+1}} (= a_k)$. 
\end{itemize}

Let $Cost_1$ denote the cost of such a solution and let $\Gamma$ be
the contribution to the cost given by colors different from
$h_0,h_1,\ldots, h_l$. Then, recalling condition (b) of the definition
of ${\mathcal I}$, we have
\begin{eqnarray*}
  Cost_1 &=& \Gamma + \sum_{w = 0, \atop w\not= k_1 }^{n-1} Q_{h_0 ,w}
  +\sum_{w = 0, \atop w\not= k_2 }^{n-1} Q_{h_1 ,w} + \cdots+\sum_{w=0, \atop w \not= k_{l+1}}^{n-1} Q_{h_l ,w}\\
         &=& \Gamma + (Q_{h_0,k_0}+Q_{h_0,k_0'}) + \cdots + (Q_{h_l,k_l}+Q_{h_l,k_l'}).
\end{eqnarray*}

Consider now a solution that differs from the previous one only by the
fact that every color in ${\cal C}_{w}$ is assigned to agent $a_w$ for $w = k_0,k_1 \ldots, k_l$.
Namely, 
\begin{itemize}
  \item $h_0\in {\cal C}_{k_0}$ is assigned to $a_{k_0}$;
  \item $h_1\in {\cal C}_{k_1}$ is assigned to $a_{k_1}$;\\
  \vdots
  \item $h_l\in {\cal C}_{k_l}$ is assigned to $a_{k_l}$. 
\end{itemize}

This is clearly a perfectly balanced solution, since each agent ``loses'' and ``gains''
exactly one color with respect to the previous case. Letting 
$Cost_2$ be the cost of such a solution, we have
\begin{eqnarray*}
  Cost_2    &=& \Gamma + \sum_{w = 0, \atop w\not= k_0 }^{n-1} Q_{h_0 ,w} +\cdots+ \sum_{w=0, \atop w \not= k_l}^{n-1} Q_{h_l ,w}\\
            &=& \Gamma + Q_{h_0,k_0'} + \cdots + Q_{h_l,k_l'}. 
\end{eqnarray*}

Hence,

$$Cost_1 - Cost_2 = Q_{h_0,k_0} + \cdots + Q_{h_l,k_l} > 0,$$ 
where the inequality follows from condition (1) of the definition of ${\cal I}$.
\end{proof}

We now consider two specific instances in ${\cal I}$ that will be used
in the following proofs.

 { For each pair $(a_i, a_{i'})$, for $i=0,\ldots, n/2-1$,} and its initially allocated set of colors
${\cal C}_i=\{i_1,\ldots,i_t\}$, fix any $u > 1$ and  
partition set ${\cal C}_i$ into subsets ${\cal C}'$ and ${\cal C}''$, each
of cardinality $t/2$.  
We define instance ${\cal I}_1 \in {\cal I}$ for the pair
$(a_i,a_{i'})$ in the following way: 

$$
\begin{array}{lll}
{\cal I}_1: & Q_{j,i}  = u & \mbox{ for each } j \in {\cal C}_i \\
  & Q_{j,i'} = Q_{j,i} = u & \mbox{ for each } j \in  {\cal C}' \\
  & Q_{j,i'} = Q_{j,i}+1 = u + 1 & \mbox{ for each } j \in  {\cal C}''
\end{array}
$$

Hence, by construction, instance ${\cal I}_1$ has the property that for any $j\in  {\cal C}''$, 
$Q_{j,i'} > Q_{j,i}$. 

\begin{example}\label{ex:uno}
Consider a pair $(a_i, a_{i'})$ with a set of colors ${\cal C}_i=\{1,2,3,4,5,6,7,8\}$. 
Let $u = 2$.
If ${\cal C}'= \{ 2,4,5,8\}$ and ${\cal C}''= \{1,3,6,7\}$,
then instance ${\cal I}_1$ will be as follows:
\begin{table}[H]
\centering
\begin{tabular}{|l|l|l|l|l|l|l|l|l|}
\hline
colors                       & 1 & 2 & 3 & 4 & 5 & 6 & 7 & 8   \\ \hline
\# items for $a_i$           & {\bf 2} & 2 & {\bf 2} & 2 & 2 & {\bf 2} & {\bf 2} & 2   \\ \hline
\# items for $a_{i'}$         & 3 & {\bf 2} & 3 & {\bf 2} & {\bf 2} & 3 & 3 & {\bf 2} \\ \hline
\end{tabular}
\end{table}
\end{example}

In the following lemma we will show that the only optimal solution to
${\cal I}_1$ is the one that assigns ${\cal C}'$ to $a_i$ and ${\cal
  C}''$ to $a_{i'}$.  The above example gives an intuition of the
formal proof. By Lemma~\ref{le:pair}, we know that only the items
{ that need to be} exchanged between $a_i$ and $a_{i'}$ account for the cost of the
optimal solution, and the latter is achieved by moving items with
weight 2 (those highlighted in bold in the table), i.e., by assigning
${\cal C}'$ to $a_i$ and ${\cal C}''$ to $a_{i'}$, { for a
  total cost of 16}.

\begin{lemma}\label{l:optuno}
  The only optimal solution to instance ${\cal I}_1$ is the one that
  assigns ${\cal C}'$ to $a_i$ and ${\cal C}''$ to $a_{i'}$.
\end{lemma}

\begin{proof}
We first compute the cost of this solution: 

$$ Cost = \sum_{j\in {\cal C}'} Q_{j,i'} +  \sum_{j\in
{\cal C}''} Q_{j,i}  = \sum_{j\in {\cal C}'} Q_{j,i} + \sum_{j\in
{\cal C}''} Q_{j,i} = \sum_{j \in {\cal C}_i}Q_{j,i} . $$

Consider any other partition of ${\cal C}_i$ into two sets
$\overline{{\cal C}'}$ and $\overline{{\cal C}''}$. Consider another
solution that assigns $\overline{{\cal C}'}$ to $a_i$ and
$\overline{{\cal C}''}$ to $a_{i'}$ and let us compute the cost of
this new solution:

\begin{eqnarray*}
  \overline{Cost} &=&
  \sum_{j\in  \overline{{\cal C}'} } Q_{j,i'}
  + \sum_{j\in  \overline{{\cal C}''} } Q_{j,i} \\
  &=&\sum_{j\in  \overline{{\cal C}'}\cap {\cal C}'}
  Q_{j,i'} + \sum_{j\in  \overline{{\cal C}'}\cap {\cal C}''} Q_{j,i'} 
  + \sum_{j\in  \overline{{\cal C}''} } Q_{j,i} \\
  &=&  \sum_{j\in  \overline{{\cal C}'}\cap {\cal C}'}
  Q_{j,i} + \sum_{j\in  \overline{{\cal C}'}\cap{\cal C}''} Q_{j,i'} + \sum_{j\in
    \overline{{\cal C}''} } Q_{j,i} \\
  &=& \sum_{j \in {\cal C}_i \setminus ( \overline{{\cal C}'}\cap {\cal
      C}'') } Q_{j,i} + \sum_{j\in  \overline{{\cal C}'}\cap {\cal C}''}
  Q_{j,i'}\\
  &>& \sum_{j \in {\cal C}_i \setminus ( \overline{{\cal C}'}\cap {\cal
      C}'') } Q_{j,i} + \sum_{j\in  \overline{{\cal C}'}\cap {\cal C}''}
  Q_{j,i} = Cost,
\end{eqnarray*}

where the inequality follows by observing that 
\begin{itemize}
\item there is at least one $j\in \overline{{\cal C}'}\cap {\cal
    C}''$, otherwise the two partitions would coincide;
 \item on instance ${\cal I}_1$ we have that for every $j \in {\cal C}''$, $Q_{j,i'} > Q_{j,i}$.
\end{itemize} 
\end{proof}

We now define the instance ${\cal I}_2 \in {\cal I}$ for the pair
$(a_i,a_{i'})$ in the following way:

$$
\begin{array}{lll}
{ {\cal I}_2}: & Q_{j,i'}  = u & \mbox{ for each } j \in {\cal C}_i \\
  & Q_{j,i} = Q_{j,i'} = u & \mbox{ for each } j \in  {\cal C}' \\
  & Q_{j,i} = Q_{j,i'} - 1 = u - 1 & \mbox{ for each } j \in  {\cal C}''
\end{array}
$$
where $\mathcal{C_i}, \mathcal{C}', \mathcal{C}''$, and $u$ are set as
before.  By construction, instance ${\cal I}_2$ has now the property
that for any $j\in {\cal C}''$, $Q_{j,i'} < Q_{j,i}$.

\begin{example}\label{ex:due}
Consider again the pair $(a_i, a_{i'})$ on the same set of colors
${\cal C}_i$  and same 
partition ${\cal C}'= \{ 2,4,5,8\}$, ${\cal C}''= \{1,3,6,7\}$, and
same $u=2 $, exactly as in Example~\ref{ex:uno}.
Instance ${\cal I}_2$ will be as follows (the cost of the optimal
solution is equal to $12$ and highlighted in bold):
\begin{table}[H]
\centering
\begin{tabular}{|l|l|l|l|l|l|l|l|l|}
\hline
colors                       & 1 & 2 & 3 & 4 & 5 & 6 & 7 & 8   \\ \hline
\# items for $a_{i}$         & 2 & {\bf 2} & 2 & {\bf 2} & {\bf 2} & 2 & 2 & {\bf 2}  \\ \hline
\# items for $a_{i'}$          & {\bf 1} & 2 & {\bf 1} & 2 & 2 & {\bf 1} & {\bf 1} & 2  \\ \hline
\end{tabular}
\end{table}
\end{example}

{ Observe that, from $a_i$ point of view, instances ${\cal
    I}_1$ and ${\cal I}_2$ are indistinguishable. Nevertheless, the
  optimal solution for instance ${\cal I}_2$ is to assign to $a_i$ the
complement set of indices with respect to the optimal solution to
instance ${\cal I}_1$.}

Analogously as the previous lemma we can prove the following result.

\begin{lemma}\label{l:optdue}
There is only one optimal solution for instance ${\cal I}_2$: assign
colors in ${\cal C}'$ to $a_{i'}$ and colors in ${\cal C}''$ to
$a_{i}$. 
\end{lemma}
\begin{proof}
The proof is very similar to that of Lemma~\ref{l:optuno}.
The cost of the solution defined in the statement is now:

$$ Cost = \sum_{j\in {\cal C}'} Q_{j,i} +  \sum_{j\in {\cal C}''} Q_{j,i'}  
        = \sum_{j\in {\cal C}'} Q_{j,i'} + \sum_{j\in {\cal C}''} Q_{j,i'} 
        = \sum_{j \in {\cal C}_i}Q_{j,i'} . $$
The cost of any other solution is calculated as in the proof of Lemma~\ref{l:optuno} 
with the exception that now instance ${\cal I}_2$ has the property that for any $j\in  {\cal C}''$, 
$Q_{j,i'} < Q_{j,i}$. 
\end{proof}

The core of the lower bound's proof lies in the simple observation
that agent $a_i$ is not able to distinguish between instance ${\cal
  I}_1$ and instance ${\cal I}_2$ without knowing also the quantities
$Q_{j,i'}$ for colors $j$ falling into partition ${\cal C}''$.

\begin{lemma}
\label{no_optimal}
If agent $a_i$ knows 
at most $t/2$ colors held by $a_{i'}$, 
it cannot compute its optimal assignment of colors.
\end{lemma}

\begin{proof} 
Construct a partition of ${\cal C}_i$ in the following way: place index
$j$ in ${\cal C}'$ if $a_i$ has knowledge of $Q_{j,i'}$ and in ${\cal C}''$ in the other case. 
If the cardinality of ${\cal C}'$ is smaller than $t/2$, arbitrarily add indices
to reach cardinality $t/2$.
Agent $a_i$ cannot distinguish between instances ${\cal I}_1$ and ${\cal I}_2$ constructed 
according to this partition of ${\cal C}_i$ and, hence, by lemmas \ref{l:optuno} and \ref{l:optdue}
cannot decide whether it is better to keep colors whose indices are in
${\cal C}'$ or in ${\cal C}''$. Finally,
observe that in both instances indices in ${\cal C}'$ are exactly
in the same position in the ordering of the colors held
by $a_{i'}$, thus the knowledge of these positions does not help.
\end{proof}

\elimina{

\begin{lemma}
\label{l:pair}
Let {\tt A} be a distributed algorithm for the problem. If during the
execution of {\tt A} there is at least one pair $(a_i,a_{i'})$ that
exchanges information about no more than $t/2$ of
their initially held colors, then the algorithm does not compute the
optimal solution.
\end{lemma}

\begin{proof}
Assume there is one pair $(a_i,a_{i'})$ for which at most $t/2 = m/n$
colors are
communicated. We show that there is an adversary
that puts $a_i$ in the situation described in Corollary~\ref{no_optimal}.
The adversary simulates communication between agents, by answering 
to $a_i$ inquires to know the $Q_{h,k}$s. 
The adversary chooses a partition in ${\cal I }$ and one between the corresponding ${\cal I}_1$ 
and ${\cal I}_2$, answering accordingly and 

{\tt cosa vuol dire la parte finale di questa frase?}

placing the first $m/n$ requests
of $Q_{j,i}$ in ${\cal C}'$. Whenever $a_i$ asks for some $Q_{j,h}$ with
$h \not= i, i'$, then the adversary answers with a fixed constant
$K$ if $j \not\in {\cal C}_i$, with zero otherwise. By Corollary~\ref{no_optimal},
it is clear that the pair $(a_i,a_{i'})$ (and thus {\tt A}) cannot compute an
optimal solution.
\end{proof}

}

\begin{theorem}
\label{msgcompl} The message complexity of the distributed color
assignment  problem on ring is $\Omega(mn)$.
\end{theorem}

\begin{proof}
  Let {\tt A} be any distributed algorithm for the problem running on
  instances in ${\cal I}$.  By the end of the execution of {\tt A},
  each agent has to determine its own assignment of colors.  Fix any
  pair $(a_i,a_{i'})$ and consider the time at which agent $a_i$
  decides its own final assignment of colors.  Assume that at
  this time $a_i$ knows information about at most 
  $t/2  = m/n$ colors of agent $a_{i'}$.  By
  Lemma~\ref{no_optimal}, it cannot determine an assignment of colors
  for itself yielding the optimal solution.

  Therefore, for all $n/2$ pairs $(a_i,a_{i'})$, agent $a_i$ has to
  get information concerning at least $m/n$ of the colors held by
  $a_{i'}$. We use Shannon's Entropy to compute the minimum number
  of bits $B$ to be exchanged between any pair $(a_i,a_{i'})$ so that
  this amount of information is known by $a_i$. We have:

$$ B = \log {{m}\choose {\frac m n }}.$$

Using Stirling's approximation and the inequality $m \geq n$, we get

\begin{eqnarray*}
B & \approx &  \frac m n  \cdot \log \frac{m\cdot n}{m+n} \\
&\geq&  \frac m n  \cdot \log \frac n 2 \in \Omega\left( \frac m n \log n \right).
\end{eqnarray*}

As a basic message contains $\log n$ bits, any pair $(a_i,a_{i'})$
needs to exchange at least $\Omega(m/n)$ basic messages. Each such
message must traverse $n/2$ links of the ring to get to one agent of
the pair to the other. As we have $n/2$ pairs of agents, the lower
bound on message complexity is given by

$$\Omega(m/n) \cdot \frac n 2 \cdot \frac n 2 \in \Omega(m\cdot n).$$

\elimina{
\item Assume that during the execution of {\tt A} at least $\Omega(n)$
  pairs  $(a_i,a_{i'})$ communicate and that the communication consists of
exchanging information about at least $m/n +1$ colors.
In this case, we need $n/2$ messages to communicate one color from
$a_i$ to $a_{i'}$ or vice-versa, hence we have at least
$n/2 \cdot (m/n +1) \cdot \Omega(n) \in \Omega(m\cdot n )$ messages.

\item Assume the hypothesis above about communication does not hold, then
there must be at least one  pair $(a_i,a_{i'})$ for which at most $m/n$
colors are  communicated. Now show that there is an adversary that  puts
$a_i$ in the situation described in Corollary~\ref{no_optimal}. 
We assume that $a_i$ asks the adversary to
know $Q_{h,k}$ (it is like the adversary simulates communication
between agents). The adversary chooses a
partition in ${\cal P}$ and answers accordingly to $I_u$ (or $I_l$),
placing the first $m/n$ requests of $Q_{j,i'}$ in ${\cal
  C}'$. Whenever $a_i$ asks for some $Q_{j,h}$ with $h\not = i,i'$,
then the adversary answers with a fixed constant $K$.

Hence, by Corollary~\ref{no_optimal}, it is 
clear that the pair $(a_i,a_{i'})$ cannot compute an optimal solution.

\end{itemize}
}
\end{proof}

\elimina{
Theorem~\ref{msgcompl} can be generalized to handle any topology:
given a graph with diameter $D$ and $m$ colors, the balanced color
assignment problem on that topology has message complexity
$\Omega(mD)$. (?)
}

\section{A distributed message-optimal algorithm}
\label{Alg}

In this section we first describe an algorithm that exhibits optimal
message complexity on synchronous ring. We will then show how to adapt
the algorithm to the case of an asynchronous ring. In the next section
we will prove that the algorithm is guaranteed to compute an
approximation of the color assignment that is within a factor three
from the optimal solution (for both synchronous and asynchronous
ring).

\subsection{Synchronous ring}


At a high level,  the algorithm consists of three phases: 
in the first phase, the algorithm elects a leader $a_0$ among the set of agents. 
The second phase of the algorithm is devoted to estimate the parameter $p
= \max_i \max _j Q_{j,i}$, {\em i.e.} the maximum number of items of a
given color held by agents.  Finally, the last phase performs the
assignment of colors to agents in such a way to be consistent with
Definition~\ref{def:BC}.  In the following we describe the three
phases in detail.

\medskip

{\noindent {\bf Algorithm {\tt Sync-Balance}}}

\medskip

\textbf{Phase 1.}  The first phase is dedicated to leader election
that can be done in $O(n)$ time with a message complexity of $O(n\log
n)$ on a ring of $n$ nodes, even when the nodes are not aware of the
size $n$ of the ring \cite{HS}. 

Leader election has also been studied in 
arbitrary wired networks \cite{GHS}.
An $O(n$ polylog$(n))$ time deterministic algorithm is available even for
ad hoc radio networks of {\em unknown and arbitrary} topology without a collision
detection mechanism, even though 
the size of the network must be known to the algorithm code 
(see \cite{CKP} for the currently best result).

Without loss of generality, in the following we will assume that agent
$a_0$ is the leader and that $a_1, a_2, ..., a_{n-1}$ are the other
agents visiting the ring clockwise. In the rest of this paper, we will
refer to agent $a_{i-1\mathrm{\ mod\ }n}$ (resp. $a_{i+1\mathrm{\ mod\
  }n}$) as to the \emph{preceding} (resp. \emph{following})
\emph{neighbor} of $a_i$.

\medskip

\textbf{Phase 2.} In this phase agents agree on an upper bound $p'$ of
$p$ such that $p' \leq 2p$.

Given any agent $a_i$ and an integer $r\geq 0$, we define:
\[ B_i(r) = \left\{ \begin{array}{ll}
    1 & \mbox{ if  $\max_{j} Q_{j,i} = 0$  and $r = 0$; } \\
    1 & \mbox{ if $2^r \leq \max_{j} Q_{j,i} < 2^{r+1}$ and $r > 0$; }\\
    0 & \mbox{ otherwise.}
                        \end{array} \right.
\]

\medskip

This phase is organized in consecutive stages labeled $0, 1, \ldots$  
At stage $r=0$, the leader sets an 
integer variable $A$ to zero, which will be updated at the end of each
stage and used to determine when to end this phase. 

In stage $r \geq 0$, agent $a_i$, for $i = 0, 1, \ldots, n-1$, waits
for $i$ time units from the beginning of the stage. At that time a
message $M$ might arrive from its preceding neighbor. If no message
arrives, then it is assumed that $M = 0$. Agent $a_i$ computes $M = M
+ B_i(r)$ and, at time unit $i+1$, sends $M$ to its following neighbor
only if $M > 0$, otherwise it remains silent.

After $n$ time slots in stage $r$, if the leader receives a message $M
\leq n$ from the preceding neighbor, then it updates variable $A = A +
M$, and, if $A < n$, proceeds to stage $r+1$ of Phase 2; otherwise
it sends a message clockwise on the ring containing the index of the
last stage $\ell$ performed in Phase 2. Each agent then computes $p' =
2^{\ell+1}$, forwards the message clockwise, waits for $n-i+1$ time units and
then proceeds to Phase 3.

\begin{lemma}
\label{l:ph2}
Phase 2 of Algorithm {\tt Sync-Balance} computes an upper bound $p'$
of $p$ such that $p' \leq 2p$ within $O(n \log p)$ time units and
using $O(n^2)$ basic messages.
\end{lemma}

\begin{proof}
  We will say that agent $a_i$ {\it speaks up} in stage $r$ when
  $B_i(r) = 1$.  Throughout the execution of the algorithm, integer
  variable $A$ records the number of agents that have spoken up so
  far.

Any agent $a_i$ speaks up in one stage only.  Indeed, given the color $j'$ for which agent 
$a_i$ has the maximum number of items, then $B_i(r) = 1$ only at stage $r$ such that $Q_{j',i}$ falls in the  
(unique) interval $[2^r, 2^{r+1})$. 
Let $a_{i^*}$ be the agent having the largest amount of items of the same color among all agents, {\em i.e.}, 
such that  $Q_{j^*,i^*} = p$, for some $j^*\in [0,m-1]$.  
Then $B_{i^*}(r) = 1$ for stage $r$ such that  $2^r \leq p < 2^{r+1}$, {\em i.e.}, agent $a_{i^*}$ 
speaks up when $r =  \ell$. Observe that at the end of stage $\ell$ the leader sets $A = n$, as all $n$ 
agents must have spoken up by that time.
Therefore, considering  also the last extra stage in which the agents are 
informed of the value of $\ell$, Phase 2 ends after $\ell+2$ stages, i.e. $n (2+ \log p)$ time units.


For what concerns message complexity, in each stage, for $r = 0,
\ldots, \log p$, either no messages are sent, or a message traverses a
portion of the ring. Observe that, as each agent speaks up only once
during this phase, messages circulating on the ring must always be
originated by different agents. Hence, the number of stages in which a
message circulates on the ring is at most $n$ and there must be at
least $\max \{ 0, \log p - n\}$ silent stages. In conclusion, Phase 2
message complexity is bounded by $O(n^2)$.

As for the ratio between the actual value of $p$ and its approximation $p'$ computed in Phase 2, 
by construction we have that $2^{\ell} \leq p$ and 

$$ p' = 2^{\ell + 1} = 2 \cdot 2^{\ell} \leq 2 p. $$

\end{proof}

\medskip

\textbf{Phase 3.}  As a preliminary step, each agent $a_i$ computes
the number of colors it will assign to itself and stores it in a
variable ${\cal K}_i$. Namely, each agent $a_i$, for
$i=0,1,\ldots,n-1$ computes $g = \left( \lfloor \frac m n \rfloor + 1
\right) n - m$ and then sets ${\cal K}_i$ as follows (recall
Definition~\ref{def:BC}):
\begin{equation}
\label{eq:Ki} 
{\cal K}_i  = \left\{ \begin{array}{ll}
                        \lfloor \frac m n \rfloor  & \mbox{ if $i < g$}
                        ; \\
                        \lfloor \frac m n \rfloor + 1 & \mbox{ otherwise.}
                        \end{array} \right.
\end{equation}
In the rest of this phase, the agents agree on a color assignment such that each agent $a_i$
has exactly ${\cal K}_i$ colors.  Algorithms~\ref{B3.1} and \ref{B3.2} report the
pseudo-code of the protocol performed by a general agent $a_i$ in this
phase and that is here described.

\medskip

Let $p'$ be the upper bound on $p$ computed in Phase 2.  Phase 3
consists of $\log p' + 1 $ stages.  In each stage $r$, for $r =
0,\ldots, \log p'$, the agents take into consideration only colors
whose weights fall in interval $I_r = [l_r, u_r)$ defined as follows:

\begin{equation}
\label{eq:intervals}
\left\{
\begin{array}{lll}
   I_0 &=& \left [\frac{p'}{2}, +\infty \right ),\\ 
   I_r &=& \left [\frac {p'}{2^{r+1}}, \frac {p'}{2^{r}} \right ) \mbox{ for } 0 < r < \log p'\\
   I_{\log p'} &=& [0, 1)
\end{array} \right.
\end{equation}

Observe that in consecutive stages, agents consider weights in decreasing order,  as
$u_{r+1} \leq l_r$. 

\medskip

At the beginning of each stage $r$, all agents have complete knowledge of
the set of colors ${\cal C}_{r-1}$ that have already been assigned to some
agent in previous stages. At the beginning of this phase, ${\cal C}_{-1}$  is the
empty set, and after the last stage is performed, ${\cal C
}_{\log p'}$ must be the set of all colors. 

Stage $r$ is, in general, composed of two steps; however, the second step
might not be performed, depending on the outcome of the first one. In the
first step, the agents determine if there is at least one agent with a
weight falling in interval $I_r$, by forwarding a message around the
ring only if one of the agents is in this situation.  If a message
circulates on the ring in step one, then all agents proceed to step
two in order to assign colors whose weight fall in interval $I_r$ and
to update the set of assigned colors.  Otherwise, step two is
skipped. Now, if there are still colors to be assigned ({\em i.e.}, if
${\cal C}_r \not= {\cal C}$), all agents proceed to stage $(r+1)$;
otherwise, the algorithm ends.  In more details:

\medskip

\noindent {\sc Step 1.}  Agent $a_i$ (leader included) waits $i$ time
units (zero for the leader) from the beginning of the stage, and then
acts according to the following protocol: 
\begin{description}
\item[Case 1:] If $a_i$ receives a
message from its preceding neighbor containing the label $k$ of some agent $a_k$, 
it simply forwards the same message to its
following neighbor and waits for $(n+k-i-1)$ time units; 
\end{description}

otherwise
\begin{description}
\item[Case 2:] If  $a_i$ has a  weight falling into interval $I_r$,
then it sends a message containing its label $i$ to its following neighbor  and
waits for $(n-1)$ time units; 
\end{description}

otherwise

\begin{description}
\item[Case 3:] It does nothing and waits for $n$ time units.  
\end{description}

If Case 1 or Case 2 occurred, then agent $a_i$ knows that {\sc step 2} is to be
performed and that it is going to start after waiting the designed time units. 

Otherwise, if Case 3 occurred, after $n$ units of time, agent $a_i$
might receive a message (containing label $k$) from its preceding
neighbor, or not.  If it does, then $a_i$ learns that Case 2 occurred
at some agent $a_{k}$ having label $k > i$ and that {\sc step 2} is to
be performed. Hence, it forwards the message to its following neighbor
in order to inform all agents having labels in the interval
$[i+1,\ldots,k-1]$, unless this interval is empty (meaning that $a_i$
was the last agent to be informed).  Then, after waiting for another
$(k-i-1)$ time units, agent $a_i$ proceeds to {\sc step 2}.  On the
contrary, if $a_i$ got no message, it learns that Case 2 did not occur
at any agent and hence, {\sc step 2} needs not be performed.  After
waiting for $(n-i)$ time units, $a_i$ can proceed to the next stage
$(r+1)$.

\medskip

Observe that, when {\sc step 2} has to be performed, {\sc step 1}
lasts exactly $n+k-1< 2n$ time units for all agents, where $k$ is the
smallest agent's label at which Case 2 occurs, while it lasts exactly
$2n$ time units for all agents in the opposite case.  Indeed,
referring to the pseudo-code in Algorithm~\ref{B3.1}, completion time
is given by the sum of the time units in the following code lines: in
Case 1 of lines 7 and 10 ($i\neq k$); in Case 2 of lines 7 and 15
($i=k$); in Case 3 of lines 7, 18 and 22 if agents proceed to {\sc
  Step 2} ($i\neq k$), and lines 7 and 26 otherwise.

As the time needed by agents to agree on skipping {\sc step 2} is
larger than the time needed to agree in performing it, it is not
possible that some agent proceeds to {\sc step 2} and some other to
stage $(r+1)$. On the contrary, agents are perfectly synchronized to
proceed to {\sc step 2} or stage $(r+1)$.

\medskip

\noindent {\sc Step 2}. When this step is performed, there exists a
non empty subset of agents having at least one weight falling into interval
$I_r$. Only these agents actively participate to the color assignment
phase, while the others just forward messages and update their list of
assigned colors. Color assignment is done using a greedy strategy:
agent $a_i$ assigns itself the colors it holds which fall into
interval $I_r$ and that have not been already assigned to other agents.
Once a color is
assigned to an agent, it will never be re-assigned to another one. 

\medskip

To agree on the assignment, the agents proceed in the following way:
agent $a_i$ creates the list ${\cal L}_{i,r}$ of colors it holds whose
weights fall into interval $I_r$ and that have not been assigned in
previous stages. Then, $a_i$ waits $i$ time units (zero for the
leader) from the beginning of the step. At that time, either $a_i$
receives a message ${\cal M}$ from its preceding neighbor or not. In
the first case, the message contains the set of colors assigned in
this stage to agents closer to the leader (obviously, this case can
never happen to the leader). Agent $a_i$ then checks if there are some
colors in its list ${\cal L}_{i,r}$ that are not contained in ${\cal
  M}$ (empty message in the case of the leader), and then assigns
itself as many such colors as possible, without violating the
constraint ${\cal K}_i$ on the maximum number of colors a single agent
might be assigned. Then, $a_i$ updates message ${\cal M}$ by adding
the colors it assigned itself, and finally sends the message to its
following neighbor.  If ${\cal L}_{i,r}$ is empty, or it contains only
already assigned colors, $a_i$ just forwards message ${\cal M}$ as it
was.  In both cases, $a_i$ then waits for a new message ${\cal M}'$
that will contain the complete list of colors assigned in this
stage. ${\cal M}'$ is used by all $a_i$ to update the list of already
assigned colors and is forwarded on the ring.  When the message is
back to the leader, stage $(r+1)$ can start.


\medskip

\begin{algorithm*}
\small
\singlespacing
\caption{{\tt Sync-Balance} - Phase 3 (performed by agent $a_i$)}
\label{B3.1}
\begin{algorithmic}[1]
\Require $p'$ computed in Phase 2 \Comment{upper bound to maximum number of
  items of the same color}
\State Compute ${\cal K}_i$  \Comment{Number of colors $a_i$ has to be assigned, as defined in Equation (\ref{eq:Ki})}
\State ${\cal C}_{-1} \gets \emptyset$  \Comment{set of colors assigned up to the  previous stage}
\For{$r=0$ {\bf to} $\log p'$}   
\State ${\cal L}_{i,r} = \{ c_j | \,\,\, c_j \not\in {\cal C}_{r-1} \mbox{ and
} Q_{j,i} \in I_r \}$ \Comment{Colors assignable to $a_i$ in stage $r$. Intervals $I_r$ are defined in (\ref{eq:intervals})}

\Statex \Comment{ { Begin of {\sc Step 1} }}
\State Wait $i$ time units
\If{Got message ${\cal M} = \{ k\}$ from its preceding neighbor} \Comment{ Case 1}
\Statex \Comment{$k<i $ is an agent label}

         \State Forward message ${\cal M}$ to its following neighbor
         \State Wait $n-i+k-1$ time units
         \State \cal{Step 2}{} \Comment{proceeds to {\sc Step 2} }
\Else 
         \If{${\cal L}_{i,r} \not= \emptyset$} \Comment{ Case 2}
         \State Send message ${\cal M} = \{ i \} $ to its following
         neighbor
         \State Wait $n-1$ time units
         \State \cal{Step 2}{} \Comment{proceeds to {\sc Step 2} }
         \Else \Comment{Case 3}
         \State Wait $n$ time units
         \If{Got message ${\cal M} = \{ k\}$ from its preceding neighbor}
                    \If{$k-i-1 > 0$} \Comment{informs other agents that Step 2 is to be performed}
                    \State Forward message ${\cal M}$  to its following
                    neighbor 
                    \State Wait for $k-i-1$
                    \State \cal{Step 2}{} \Comment{procedure call to {\sc Step 2} }
                    \EndIf   
                    \Else
                    \State Wait $n-i$ time units \Comment{proceeds to next stage skipping
           {\sc Step 2}} 
         \EndIf 
         \EndIf
\EndIf
\EndFor
\end{algorithmic}
\end{algorithm*}

\begin{lemma}\label{l:worst}
Let $K_r$ be the number of colors assigned in stage $r$ of Phase 3, then stage $r$ can be completed in at most 
$O(n)$ 
time units using at most 
$O\left(n\cdot \frac {K_r\log m}{\log n} \right)$ 
basic messages. 
\end{lemma}

\begin{proof}
  The bound on the time complexity follows straightforwardly by
  observing that each of the two steps requires at most $2n$ time
  units.

\medskip

For what concerns message complexity, {\sc Step 1} requires no
messages if {\sc Step 2} is skipped, and $n-1$ otherwise.  In fact,
only one basic message goes clockwise on the ring from $a_k$ to
$a_{k-1}$, where $k$ is the smallest index at which Case 2 occurs. The
worst case for {\sc Step 2} is the case in which the leader itself
assigns some colors, as a possibly long message containing color ID's
must go twice around the ring.  As there are $m$ colors, one color can
be codified using $\log m$ bits, then, sending $K_r$ colors requires
no more than $\frac {K_r\log m}{\log n}$ basic messages.  In
conclusion, the total number of basic messages is upper bounded by
$O\left(n\cdot \frac {K_r\log m}{\log n} \right)$.
\end{proof}

\begin{algorithm*}
\small
\singlespacing
\caption{{\tt Sync-Balance} - Phase 3  {\sc Step 2} (performed by agent $a_i$)}
\label{B3.2}
\begin{algorithmic}[1]

\Procedure{Step 2}{}
\State Wait $i$ time units
\If{Got message ${\cal M} $ from preceding neighbor with
  list of colors}
          \State ${\cal L}_{i,r} \gets {\cal L}_{i,r} \setminus
          {\cal M} $ \Comment{list of candidate colors to self assign}
          \Else
          \State Create empty message ${\cal M}$
          \EndIf
 \If{$|{\cal L}_{i,r} | \not= \emptyset$}
          \State Self assign maximum number of colors among those in
          ${\cal L}_{i,r}$ 
          \Statex \Comment{the total number of colors $a_i$ can assign itself is given by ${\cal K}_i$}
          \State Add self assigned colors to ${\cal M} $
          \EndIf
\If{${\cal M} \not= \emptyset$ }
          \State Send message ${\cal M}$  to the following neighbor
          \EndIf
\State Wait for message ${\cal M}'$ from preceding neighbor with
  list of colors
\State $C_r \gets C_{r-1} \cup {\cal M}'$ \Comment{updates set of
  assigned colors}
\State Forward message ${\cal M}'$  to the following neighbor
\If{${\cal C}_r = {\cal C}$} \Comment{all colors have been assigned}
        \State {\bf stop}
        \Else
        \State Wait for $n-i$ time units
        \EndIf
\EndProcedure
\end{algorithmic}
\end{algorithm*}

\begin{corollary}
\label{cor:ph3}
Phase 3 of Algorithm {\tt Sync-Balance} can be completed within $O(n\log p)$ 
time units and  using  $O(n m \cdot\frac{\log m }{ \log n})$ basic messages.
\end{corollary}

\begin{proof}
It will suffice to sum up the worst cases for message and time complexity 
from Lemma~\ref{l:worst} over all stages 
$r=0, \ldots, \log p'$, where $p' \leq 2p$ (Lemma~\ref{l:ph2}). 

The upper bound on the time complexity is straightforward.
Let $K_r$ be defined as in the statement of Lemma~\ref{l:worst}, i.e. as
the number of colors assigned in a generic stage $r$ of Phase 3.
The upper bound on the message complexity follows by observing that 
$\sum_{r=0}^{\log p' } K_r = m $, as the total number of assigned colors 
during the $\log p' +1 $ stages is exactly the given number of colors. 

\end{proof}


We are now ready to prove that our algorithm is correct. In
Section~\ref{Approx} we will evaluate the ratio of the cost of the
solution found by this algorithm and the one of the optimal solution.

\begin{theorem}
Assuming $m \in O(n^c)$, for some constant $c$, Algorithm {\tt Sync-Balance} finds a feasible solution to the balanced
color assignment  problem in time $O(n\log p)$ using $\Theta(mn)$ messages.
\end{theorem}

\begin{proof}  
To prove correctness, we show that any assignment of  colors to
agents computed by algorithm {\tt Sync-Balance} satisfies the two 
following conditions: 
\begin{itemize}
\item[$(i)$]  A color $c_j$ cannot be assigned to more than one agent.
\item[$(ii)$] All colors are assigned.
\end{itemize}


\noindent $(i)$ The algorithm can assign a new color $c_j$ to agent $a_i$ only in line 9 of Algorithm~\ref{B3.2}.
This can only happen if $c_j$ has not been already assigned in a previous stage, or in the current stage to an 
agent with smaller label. Since, in the stage, the color assignment is done sequentially (starting from the leader
 and following the ring
clockwise), no color can be assigned to two different agents. Moreover, in lines 15-17 of Algorithm~\ref{B3.2}, 
all agents update the list of colors assigned in the current stage and, hence, in later stages, 
already assigned colors will not be assigned again.  Therefore {\tt Sync-Balance} prevents the
assignment of the same color to two different agents.

\medskip

\noindent $(ii)$ If an available color $c_j$ of weight
$Q_{j,i}\in\lbrack l_r,u_r)$ is not taken by $a_i$ during
stage $r$, it is only because $a_i$ has enough colors already
(line 9). However, this circumstance may not occur at all
agents during the same stage (for this would imply that there were
more than $m$ colors). Thus, either the color is taken by a higher
labeled agent in stage $r$, or is ``left free'' for agents for
which the weight of $c_j$ is less than $l_r$. By iterating the
reasoning, we may conclude that, if not taken before, the color
must be eventually assigned in stage $\lceil\log p\rceil+1$, where
agents are allowed to pick colors for which their weight is zero.

\medskip

As for upper bounds on time and message complexities, by summing up upper bounds for  the three phases, we have 

$$
\begin{array}{ll}
\mbox{ Time complexity:} & \underbrace{O(n)}_\text{Phase 1} + \underbrace{O(n\log p)}_\text{Phase 2} + \underbrace{O(n\log p)}_\text{Phase 3} = O(n\log p),\\
 & \\
\mbox{ Message complexity:} &  \underbrace{O(n\log n)}_\text{Phase 1} + \underbrace{O(n^2)}_\text{Phase 2} 
+ \underbrace{O\left(n m \cdot\frac{\log m }{ \log n}\right)}_\text{Phase 3} = O(nm),
\end{array}
$$

where we used Lemma~\ref{l:ph2}, Corollary~\ref{cor:ph3}, and the facts that $m\geq n$ and that 
$\log m / \log n \in O(1)$, under 
the given hypothesis. 


\end{proof}

\subsection{Asynchronous ring}

In an asynchronous ring such instructions as {\em "wait for i time
  units"} (see Algorithm~\ref{B3.1} and \ref{B3.2}) cannot guarantee a
correct completion of the global algorithm. Here we discuss how to
make simple modifications to {\tt Sync-Balance} in order to get an
algorithm (named {\tt Async-Balance}) that correctly works in the
asynchronous case as well.

The leader election in Phase 1 can be done in $O(n)$ time with a
message complexity of $O(n\log n)$ even on an asynchronous ring of $n$
nodes \cite{HS}.  Therefore, the main differences are in Phase 2 and
Phase 3.
 
In Phase 2 we propose a slightly different strategy that works in only
2 stages, instead of $\log p$. This better time complexity translates,
in general, into an extra cost in terms of message complexity.
Nevertheless, under reasonable hypothesis (namely when $p \in
O(m^{m})$), the message complexity reduces to the same bound as for
the synchronous setting.

Finally, in Phase 3, the main ideas remain the same, but there are no
``silent stages'' and the leader acts differently from the other
agents, as it is the one originating all messages circulating on the
ring.

In the following we highlight the main differences with the synchronous protocol: 

\medskip

{\noindent {\bf Algorithm {\tt Async-Balance}}}

\medskip

\noindent {\bf Phase 1.} Leader election can be accomplished with an $O(n \log n)$ message complexity~\cite{HS}. 

\noindent {\bf Phase 2.} This phase consists of only two stages.  In
the first stage the agents compute $p = \max_i \max_{j} Q_{j,i} $.
Let $p_i = \max_j Q_{j,i}$, {\em i.e.} the maximum number of items of
the same color agent $a_i$ posseses.  The leader originates a message
containing $p_0$. Upon reception of a message $M$ from its preceding
neighbor, agent $a_i$ computes $M = \max \{ M, p_i \}$ and forwards
$M$ to its following neighbor.  The message that gets back to the
leader contains $p$ and it is forwarded once again on the ring to
inform all agents.

\elimina{ At each stage $r$, the leader generates a message $M =
  B_0(r)$ that is sent clockwise on the ring. Agents $a_i$ that
  receives the message, computes a new message $M = M + B_i (r)$ and
  forwards it on the ring again. When the message gets back to the
  leader, variable $A$ is updated as in the synchronous case, and if
  $A = n$ the leader communicates to all agents the end of the
  phase. - non va perche' girano n messagi per alla peggio $\log p$
  stage.  }

Observe that Phase 2 requires no more than 
$O\left(n\cdot \frac{\log p}{\log n} \right)$ 
basic messages, as $O( \log p / \log n) $ basic messages
are needed to send the $p_i$'s and $p$.

\noindent {\bf Phase 3.} Changes in this phase concern both the
execution of {\sc step 1} and {\sc step 2}, that are to be modified in
the following way:

\medskip

\noindent {\sc Step 1}. Each agent $a_i$ computes its list of
assignable colors ${\cal L}_{i,r}$ and sets $Y_i(r) = 1$ if $|{\cal
  L}_{i,r}| > 0$, and $Y_i(r) = 0$ otherwise.  The leader starts the
step by sending, to its following neighbor, a basic {
  boolean} message containing
$Y_0(r)$. Upon reception of a message $M$ from its preceding neighbor,
agent $a_i$ computes $M = M \lor Y_i(r)$ and forwards $M$ to its
following neighbor.  When the leader gets the message back, it
forwards the message again on the ring, and the same is done by all
agents, until the message arrives to $a_{n-1}$. The second time one
agent (leader included) gets the message, it checks its content: if it
is a one, then it knows that it has to proceed to {\sc Step 2};
otherwise, if it contains a zero, it proceeds to the next
stage. 

\medskip

\noindent {\sc Step 2}. The leader starts the step by sending, to its
following neighbor, a (possibly empty) list of self assigned colors,
obtained exactly as in the synchronous case. Then agents act as in the
synchronous protocol, with the exception that they are activated by
the arrival of a message from the preceding neighbor and not by a time
stamp. Agents proceed to the next stage after forwarding the complete
list of colors assigned in the stage.  \elimina{ This message is sent
  right after the last one of {\sc Step 1}. Any other agent waits for
  a message coming from its preceding neighbor, that will surely
  arrive. Then, it self assigns some colors, if it has to, and
  consequently updates the incoming message. The modified message, or
  the original one if no assignment is done, is then forwarded to its
  following neighbor. When the message gets to the leader again, this
  is forwarded once more around the ring, so that all agents are able
  to update the list of assigned colors. The last agent on the ring
  ($a_{n-1}$) stops the message. Agents proceeds to the next stage on
  the time unit successive to message forwarding. }

\medskip

\begin{lemma}
  Let $K_r$ be the number of colors assigned in stage $r$ of Phase 3,
  then stage $r$ can be completed using at most 
   $O\left( n\cdot \frac{K_r\log m}{\log n}\right)$ 
  basic messages.
\end{lemma}

\begin{proof} {\sc step 1} is always performed and a basic message is
  forwarded (almost\footnote{On the second stage, agent $a_{n-1}$
    stops the message.}) twice around the ring. Hence, $O(n)$ basic
  messages are used. When {\sc step 2} is performed, a message
  containing color ID's goes (almost) twice around the
  ring. Analogously to the synchronous case, we can prove that no more
  than 
  $O\left(n\cdot \frac{K_r \log m}{ \log n}\right)$
  messages are needed.
\end{proof}

Analogously to the synchronous case, we can prove the following corollary.

\begin{corollary}
Phase 3 of Algorithm {\tt Async-Balance} can be completed using  $O\left(n m\cdot \frac{\log m }{ \log n}\right)$ basic messages.
\end{corollary}

\elimina{
\begin{algorithm*}
\small
\singlespacing
\caption{{\tt Async-Balance} - Phase 3  {\sc Step 1} (performed by agent $a_i$)}
\label{AB3.1}
\begin{algorithmic}[1]

\Statex .... \Comment{As in {\tt Sync-Balance}} 
\For{$r=0$ {\bf to} $\lceil\log p' \rceil$}   
\Statex .... \Comment{As in {\tt Sync-Balance}}
\Statex {\sc Step 1} 
 \If{$i=0$}  \Comment{Agent $a_i$ is the leader}
	\If{${\cal L}_{0,r} = \emptyset$}
		\State Send message ${\cal M} = \{ 0 \}$ to following neighbor
		\Else 
		\State Send message ${\cal M} = \{ 1 \}$ to following neighbor
		\EndIf
	\State \cal{Decide}{}	
\EndIf
\If{$i>0$}  \Comment{Agent $a_i$ is not the leader}
	\State Wait for message ${\cal M}$ from following neighbor
	\If{${\cal L}_{i,r} = \emptyset$}
		\State Forward message ${\cal M} $ to following neighbor
		\Else 
		\State Send message ${\cal M} = \{ 1 \}$ to following neighbor
		\EndIf
	\State \cal{Decide}{}	
\EndIf
\EndFor

\Statex

\Procedure{Decide}{}  \Comment{Makes agents decide if to go to {\sc Step 2} or next stage}
\State Wait for message ${\cal M}$ from following neighbor
	\If{$i \not= n-1$}
	\State Forward ${\cal M}$ to following neighbor
	\EndIf
	\If{${\cal M} = \{ 1 \}$}
		\State \cal{Step 2}{}  
	\EndIf   \Comment{otherwise,  goes to next stage}
\EndProcedure

\end{algorithmic}
\end{algorithm*}

\elimina{
\begin{algorithm*}
\small
\singlespacing
\caption{{\tt Async-Balance} - Phase 3  {\sc Step 2} (performed by agent $a_i$)}
\label{AB3.1}
\begin{algorithmic}[1]
\Statex Forse non serve, visto che e' molto piu' simile all'altro e facile?
\end{algorithmic}
\end{algorithm*}
}
}

\begin{theorem}
  Assuming $m \in O(n^c)$, for some constant $c$, Algorithm {\tt
    Async-Balance} finds a feasible solution to the balanced color
  assignment problem, on asynchronous rings, within time $O(n\log p)$
  using $O\left(n\cdot \frac{\log p}{ \log n} + nm \right)$ basic messages.
\end{theorem}

\begin{proof}
The correctness proof is analogous to the synchronous case. 

The time complexity is asymptotically equivalent to the synchronous
case. Indeed, as already mentioned, the leader election in Phase 1 can
be completed in $O(n)$ time, Phase 2 requires 2 circles around the
ring and, finally, Phase 3 includes $O(\log p)$ stages, each of them
requiring 2 circles around the ring.

For what concerns message complexity, summing up upper bounds for
single phases, we get

$$ O( n \log n) + O\left( n\cdot \frac{\log p}{ \log n} \right) + 
O\left(n m\cdot \frac{\log m }{ \log n}\right) = O\left(n\cdot \frac{\log p}{ \log n} + nm \right), $$

as $m \in O(n^c)$.

\end{proof}

When we also have that $O(\log p) = O(m \log m)$, the algorithm
exhibits the same {\it optimal} message complexity as in the
synchronous setting. Namely, we can state the following result.

\begin{corollary}
  \label{c19}
  If $m \in O(n^c)$, for some constant $c$, and $p \in O(m^{m})$, then
  Algorithm {\tt Async-Balance} finds a feasible solution to the
  balanced color assignment problem, on asynchronous rings, using
  $\Theta (nm)$ messages.

\end{corollary}

\elimina{
{ \tt sezione da controllare }

\subsubsection{About using the asynchronous protocol in the synchronous case}

Observe that the asynchronous protocol obviously works also in the
synchronous case, but it involves the expedition of a larger number of
messages. Even if this does not show in the worst case asymptotic
message complexity, this second protocol always uses $2n-1$ messages
in {\sc Step 1}, even if {\sc Step 2} is not to be performed. On the
contrary, the protocol presented for the synchronous case uses only
$n$ messages in the case agents have to proceed to {\sc Step 2} and no
message at all in the opposite case. As for {\sc Step 2}, let $k$ be
the smallest agent label such that ${\cal L}_{i,r} \setminus {\cal M}
\not= \emptyset$, then the protocol for the asynchronous case requires
extra $0\leq k \leq n-1$ messages with respect to the synchronous
case.

On the other side, the asynchronous protocol is faster on a
synchronous ring, as steps and stages might be highly
overlapped. According to the specific setting in which the algorithm
is applied, and depending on which resource, between time and
messages, is the most expensive, one can decide to use one protocol or
the other (obviously only in the case of a synchronous ring).

\begin{lemma}
{\bf Phase 3} of {\tt Async-Balance} applied to a synchronous ring can be completed within $2(n +1)\log p' $ time units. 
\end{lemma}

\begin{proof}
For any stage $r = 0, \ldots, \log p' $, we have that
{\sc Step 1} for the leader $a_0$ ends after $(n+1)$ time units, and, when performed, {\sc Step 2} ends after other  $(n+1)$ time units. The leader starts the next stage on the successive time unit (making two stages overlap). 
\end{proof}

Using the synchronous protocol for leader election, we have: 

\begin{theorem}
Assuming $m \in O(n^c)$, for some constant $c$, Algorithm {\tt Async-Balance} finds a feasible solution to the balanced
color assignment  problem, on synchronous rings, within $ O(n\log p)$ time units. 
\end{theorem}

As we said, even if it is not evident from worst case asymptotic computational cost, there is a multiplicative factor of 2 between time and message complexity in using the asynchronous protocol for a synchronous ring: the protocol is twice as fast, but uses a doubled number of messages.  {\tt controllare affermazione}


\elimina{
We now formally describe the algorithm performed by any agent $a_i$ at
any stage $r$. Message ${\cal M}$ is received from the right neighbor
and contains a list of the colors already assigned to agents closer to
the leader at the present stage. ${\cal M}'$ is the list of candidate
colors to be assigned to agent $a_i$: colors whose weights fall in the
interval relative to the stage $r$, that have not been already
assigned to agents closer to the leader in stage $r$ or to any other
agent in previous stages ({\em i.e.}, ${\cal M}' \cap ({\cal M} \cup
{\cal C}_{r-1}) = \emptyset$). Finally, ${\cal M}^* \subseteq {\cal
M}'$ contains the maximum number of colors still assignable to the
agent ({\em i.e.} $|{\cal M}^*| \leq
\frac{m}{n}-|\xi_i|$ ). \\

Any agent $a_i$ performs the following protocol.
{
\noindent 
\begin{itemize}
\item[1.]  ${\cal C}_{-1}:=\xi_i:=\emptyset;$
\item[] {\bf For} $r:=0$ {\bf to} $\lceil\log p\rceil$ {\bf do}
   \begin{itemize}
   \item[2.]  Wait for message ${\cal M}$ from the right;   \,\,\,\,\,\,\,\,\,/*
       { the leader starts the stage by setting} ${\cal M}:=\emptyset;$\,\,\,*/
   \item[3.] ${\cal M}':=\{c_j|\,\,\, c_j\notin{\cal M}\cup{\cal C}_{r-1}\mbox{ and }
              Q_{j,i}\in\lbrack l_r,u_r) \};$\,\,
    /* {if $r=0$, \,\,\,$Q_{j,i}\in\lbrack l_r,u_r\rbrack$}\,\,\,*/
    \item[4.] Let ${\cal M}^*$ be the set with the
 $\min\lbrace\frac{m}{n}-|\xi_i|,|{\cal M}'|\rbrace$
              colors of highest weight in ${\cal M}'$;
    \item[5.] $\xi_i:=\xi_i \cup {\cal M}^*$;  \,\,\,\,\,\,\,\,\,/*
              {$a_i$ updates its own set of assigned colors} */
    \item[6.] Send message ${\cal M} \cup {\cal M}^*$ to the left;
    \item[7.] Wait for message ${\cal M}_r$ from the right;
\,\,\,\,\,\,\,\,\,
    \item[8.] ${\cal C}_r:={\cal C}_{r-1} \cup  {\cal M}_r$;
 \,\,\,\,\,\,\,\,\,/*
               {Update ${\cal C}_r$ with the colors assigned in
 stage} $r$ */
    \item[9.] Send message ${\cal M}_r$ to the left;
    \item[10.] If ${\cal C}_r = {\cal C}$ then {\bf stop};
    \end{itemize}
 \item[] {\bf nextfor}
 \end{itemize}
}
}

}

\section{Approximation Factor of Algorithm {\tt Balance} }
\label{Approx}

The main result of this section is that the cost of the solution (as
defined in Definition~\ref{def:cost}) computed by the algorithms
presented in the previous sections is at most a small {\it constant}
factor larger than the cost of the optimal solution.  Namely, we will
show that it is at most three times the optimal solution and that the
analysis is tight.  Moreover, we will show how to modify the algorithm
to get a $2$-approximation ratio at the expenses of a little increase
of message complexity, and, for the synchronous case only, how to get
a $(2+ \epsilon)$-approximation ratio (for every $0<\epsilon <1$) at
the expenses of an increase in time complexity.

Since, under the { same assumptions} of Corollary \ref{c19}, the cost of the
solution is the same both in the synchronous and asynchronous versions
(the assignment of colors is exactly the same in both cases), in this
section we will address both {\tt Sync-Balance} and {\tt
  Async-Balance} with the generic name {\tt Balance}. In the following
some results are expressed in terms of the value $p'$ (respectively,
$p$) computed by the agents in the synchronous (resp. asynchronous)
case during Phase 2 of the algorithm. As these results hold for both
$p'$ and $p$, to avoid repeating the distinction between $p'$ and $p$
over and over again, we will indicate with ${\hat p}$ both values $p'$
and $p$.

We begin with the following lemma:

\begin{lemma}
\label{stessostage}
Let color $c_j$ be assigned to agent $a_i$ in stage $r$ (of Phase 3)
by algorithm {\tt Balance}.  Let $a_{k}$ be a different agent such
that $Q_{j,k}\in\lbrack l_r,u_r)$.  Then $Q_{j,i} \leq 2 \cdot
Q_{j,k}.$
\end{lemma}

\begin{proof} 
If $r=\lceil \log {\hat p} \rceil + 1$ ({\em i.e.}, is the last stage), then
it must be  $Q_{j,i} = Q_{j,k}= 0$, and we
are done. Otherwise,  as  $c_j$ is assigned to agent $a_i$ in stage
$r$ then it must be $Q_{j,i}\in\lbrack l_r,u_r)$ and   the thesis easily follows from

\[
  {\frac {{\hat p}}{2^{r+1}}}  \leq
  Q_{j,i} , Q_{j,k} <  {\frac{{\hat p}}{2^r}}.
\]
\end{proof}

Let $\piB: \{ 1,\ldots, m\} \rightarrow \{1,\ldots, n\}$ be the assignment of colors to agents determined by
algorithm {\tt Balance}, and let $\piO: \{ 1,\ldots, m\} \rightarrow \{1,\ldots, n\}$ be an optimal
assignment. Define a partition of the set of colors based on their
indices, as follows:

\begin{itemize}
\item ${\cal C'} = \{ j \mid \piB (j) = \piO (j) \}$; {\em i.e.},
color indices for which the assignment made by algorithm
{\tt Balance} coincides with (that of) the optimal solution.
\item ${\cal C''} = \{0,\ldots ,m-1 \} \setminus {\cal C'}$; {\em i.e.},
colors indices for which the assignment made by algorithm
{\tt Balance} is different from the one of the optimal solution.
\end{itemize}

\begin{lemma}
  \label{massimo} Assume $\mathcal{C}''$ is not empty (for otherwise
  the assignment computed by \texttt{Balance} would be optimal) and
  let $j \in {\cal C''}$. Let $k\neq j$ be any other color index
  in ${\cal C''}$ such that $\piB (k) = \piO(j)$.  Then $$ Q_{j,\piO
    (j)} \leq \max \{ 2\cdot Q_{j,\piB (j)} , Q_{k,\piB (k)} \}.$$
\end{lemma}

\begin{proof} { First observe that, 
\elimina{if $j \in {\cal C''}$, then by a
  simple cardinality argument, there must exist at least a distinct
  index $k\in {\cal C''}$ such that $\piB (k) = \piO(j)$.  Moreover,}
  as $j,k \in {\cal C''}$ and $\piB (k) = \piO(j) \not=
  \piB(j)$, we have that $\piB(j) \not= \piB (k)$.}

If $ Q_{j,\piO (j)} \leq Q_{k,\piB (k)} $ we are clearly done. 
Suppose now that $Q_{j,\piO (j)} > Q_{k,\piB (k)}$, then we can prove that 
$Q_{j,\piO (j)} \leq  2\cdot Q_{j,\piB (j)}$. 

The fact that $j \in {\cal C''}$ means that {\tt Balance} assigned
color $c_j$ to a different agent compared to the assignment of the
optimal solution. Let $r$ be the stage of {\tt Balance} execution in
which agent $\piO (j)$ processed color $c_j$ ({\em i.e.}, $Q_{j,\piO
  (j)} \in I_r$) and could not self assign $c_j$, then (in principle)
one of the following conditions was true at stage $r$:

\begin{enumerate}
\item $\piO (j)$ already reached its maximum number of colors before
  stage $r$. \\
  However, this is impossible. It is in fact a contradiction that
  $\piO (j)$ gets color $c_k$ (recall that $\piO (j)= \piB (k)$) but
  does not get color $c_j$ under \texttt{Balance}, since we are
  assuming $Q_{j,\piO (j)} > Q_{k,\piO (j)}$, which means that the
  assignment of $c_k$ cannot be done earlier than $c_j$'s assignment.
\item Color $c_j$ has already been assigned to $\piB (j)$. 
 This might happen because
\begin{enumerate}
\item $c_j$ has been assigned to $\piB (j)$ in a previous stage. \\
This implies that $\piB (j)$ has a larger number of items of color
$c_j$ with respect to $\piO (j)$, {\em i.e.}, that $Q_{j,\piO (j)}
\leq  Q_{j,\piB (j)} \leq  2\cdot Q_{j,\piB (j)}$.
\item $c_j$ has been assigned to $\piB (j)$ in the same stage, because
  it has a smaller label on the ring.\\
By Lemma~\ref{stessostage} we then have that $Q_{j,\piO (j)} \leq  2\cdot Q_{j,\piB (j)}$.
\end{enumerate} 
\end{enumerate}
 
\end{proof}

\begin{theorem}
{\tt Balance} is a $3$-approximation algorithm for the Distributed Balanced  Color
Assignment Problem.
\end{theorem}
\label{3approx}

\begin{proof} Let $\costB$ and $\costO$ be the cost of the solutions
given by algorithm {\tt Balance} and
$OPT$, respectively. We can express these costs in the following way
(where, for simplicity, we omit index $i$'s range, that is always
$[0,n-1]$):

\begin{eqnarray*}
\costB &=& \sum_{j=0}^{m-1} \sum_{i\not= \piB (j)} Q_{j,i} \\
&=&\sum_{j\in {\cal C'} }
\sum_{ i\not= \piB (j)} Q_{j,i} +
\sum_{j\in {\cal C''}} \sum_{i\not= \piB (j)} Q_{j,i} \\ &=&
\sum_{j\in {\cal C'} } \sum_{ i\not= \piB (j)} Q_{j,i} +
\sum_{j\in {\cal C''}} \left( Q_{j,\piO (j)} + \sum_{i  \not=
\piO (j)  , \atop i\not=\piB (j)}   Q_{j,i} \right).
\end{eqnarray*}
Analogously,
\begin{eqnarray*}
\costO &=&
\sum_{j\in {\cal C'} } \sum_{ i\not= \piO (j)} Q_{j,i} +
\sum_{j\in {\cal C''}} \left(Q_{j,\piB (j)} + \sum_{i  \not= \piO
(j)  , \atop i\not=\piB (j)}   Q_{j,i} \right)
\end{eqnarray*}
By definition, $\piB (j)= \piO (j)$, for $j\in {\cal C'}$, and
thus $\sum\limits_{j\in {\cal
C'} } \sum\limits_{ i\not= \piB (j)} Q_{j,i} =
\sum\limits_{j\in {\cal C'} } \sum\limits_{ i\not= \piO (j)}
Q_{j,i}$, {\em i.e.}, the cost associated with color $c_j \in
{\cal C'}$ is exactly the same for {\tt Balance} and $OPT$. Notice also
that the term $\sum\limits_{j\in {\cal C''}} \sum\limits_{i  \not=
\piO (j) \atop i\not=     \piB (j)}   Q_{j,i} $ appears in both
cost expressions. Hence, to prove that  $\costB \leq 3 \cdot  \costO$,
it is sufficient to show that 

\begin{equation}\label{eq:diff}
  \sum\limits_{j\in {\cal C''}} Q_{j,\piO (j)} \leq 3 \sum\limits_{j\in {\cal C''}} Q_{j,\piB  (j)}.
\end{equation}

{

We can assume without loss of generality that $m$ is a multiple of $n$.
Indeed, if otherwise $n$ does not divide $m$, we can add $r$ dummy colors (for $r = m-\lfloor m/n \rfloor$), 
i.e. such that $Q_{j,i} = 0$ for all agents $i$ and dummy color $j$.
Since in our algorithm the agents consider the weights in decreasing order, the dummy colors will be processed
at the end and therefore they have no effect on the assignment of the other colors. Moreover, as their
weights are zero, they do not cause any change in the cost of the solution.

To prove (\ref{eq:diff}), we build a partition of the set ${\cal C''}$ according to the following procedure.  
We start from any $j_1$ in ${\cal C''}$ and find another index $j_2$
such that $\piB(j_2) = \piO(j')$ for some $j'\in {\cal C''}\setminus \{j_2\}$. 
Note that, since $m$ is a multiple of $n$, every agent must have $m/n$ colors and therefore such an index $j_2$
must exist.
If $j'=j_1$ the procedure ends, 
otherwise we have found another index $j_3 = j'$ such that $\piB(j_3) = \piO(j'')$.
Again, if $j''=j_1$ the procedure ends, otherwise
we repeat until, for some { $t\geq 2$}, we eventually get 
{ $\piB(j_t) = \piO(j_1)$ }. 
We then set 

$$ {\cal C}_1 = \left\lbrace (j_1,j_2), (j_2,j_3), \dots,{(j_{t-1},j_t)} \right\rbrace.$$
}
 
If during this procedure we considered all indices in ${\cal C''}$ we stop, otherwise, we pick another index
not appearing  in $ {\cal C}_1$ and repeat the same procedure to define a
second set ${\cal C}_2$, and so on until each index of ${\cal C''}$
appears in one ${\cal C}_i$. Observe that each ${\cal C}_i$ contains
at least two pairs of indices and that each index $j \in {\cal C''}$
appears in exactly two pairs of exactly one ${\cal C}_i$. 



Then, using  Lemma~\ref{massimo}, we get

\begin{eqnarray*}
  \sum\limits_{j\in {\cal C''}} Q_{j,\piO (j)} &=& \sum\limits_{{\cal
      C}_i} \sum\limits_{(j,j') \in {\cal C}_i} Q_{j,\piO (j)}  \\
  &\leq& \sum\limits_{{\cal C}_i} \sum\limits_{(j,j') \in {\cal C}_i}
  \max \{ 2\cdot Q_{j,\piB (j)} , Q_{j',\piB (j')} \} \\
  &\leq& \sum\limits_{{\cal C}_i} \sum\limits_{(j,j') \in {\cal C}_i}
  \left( 2\cdot Q_{j,\piB (j)} + Q_{j',\piB (j')} \right) \\
  &=& \sum\limits_{j\in {\cal C''}} 3 Q_{j,\piB (j)}
\end{eqnarray*}
\end{proof}

The following theorem shows that the approximation factor given in
Theorem~\ref{3approx} is tight. 

\begin{theorem}\label{epsilon}
 For any $0<\epsilon<1$, there exist instances of the Balanced Color Assignment
 Problem such that $COST_B$ is a factor $3-4\epsilon/(4\delta +\epsilon)$
 larger than the  optimal cost, for some $0 <\delta < 1$.
\end{theorem}

\begin{proof} Consider the following instance of the balanced color assignment
problem. For the sake of presentation, we assume that  $m=n$ and that
$n$ is even, but it is straightforward to extend the proof to the
general case. 

Fix any rational $\epsilon>0$, and let $q,\delta >0$ be such that
$q\epsilon/4$ is an integer and $q\delta = \lfloor q \rfloor$. 
Consider an instance of the problem such that colors are distributed as follows:

\[ \left\{ \begin{array}{lcl}
                          Q_{2i,{2i}}&=&q(\delta+\epsilon/4)  \\
                          Q_{2i+1,2i}&=&q\\
 Q_{2i,2i+1}&=&q(2\delta -\epsilon/4)\\
 Q_{2i+1,2i+1}&=&0,
  \end{array}
   \right.
\,\,\,\,\,\,\,  i=0,1,\ldots,\frac{n}{2}-1
\]
and that $a_0$ is the leader elected in the first stage of
algorithm {\tt Balance}, and that the labels assigned to agents
$a_1,\ldots,a_{n-1}$ are $1,\ldots ,n-1$, respectively. 

Consider agents $a_{2i}$ and $a_{2i+1}$, for any $0\leq i\leq
\frac{n}{2}-1$. We can always assume that $q$ is such that 

$$\frac {{\hat p}}{2^{r+1}}\leq q \delta < q(\delta+\epsilon/4)< q(2\delta -\epsilon/4) < {\frac
{{\hat p}}{2^r}},$$ 

for some $r$. That is, the weights
of color $c_{2i}$ for agents $a_{2i}$ and $a_{2i+1}$ belong to the
same interval $\left[ {\hat p}/2^{r+1}, {\hat p}/{2^r} \right)$.

It is easy to see that the optimal assignment gives $c_{2i+1}$ to
$a_{2i}$ and $c_{2i}$ to $a_{2i+1}$. The corresponding cost is
$\costO=\frac{n}{2}q(\delta+\epsilon/4)$. On the other hand, algorithm
{\tt Balance} assigns $c_{2i}$ to $a_{2i}$ and $c_{2i+1}$ to $a_{2i+1}$,
with a corresponding cost
$\costB=\frac{n}{2}q(3\delta-\epsilon/4)$. Hence, 
for the approximation factor, we get

\[ \frac{\costB}{\costO} = \frac{3\delta
  -\frac{\epsilon}{4}}{\delta+\frac{\epsilon}{4}}
 = \frac{3\left(\delta + \frac{\epsilon}{4} \right)}{\delta +\frac{\epsilon}{4} } 
- \frac{\frac{3\epsilon}{4} + \frac{\epsilon}{4}}{\delta +\frac{\epsilon}{4} } =  
3 -\frac{4\epsilon}{4\delta +\epsilon} .\]  
\end{proof}

\elimina{

\begin{figure}[htbp]
\label{fig:ring6}
\centerline{\epsfig{file=figura_nn.eps, width=5in}}
\caption{ Example. Six agents ring 
  with initial color distribution: $Q_{2i,2i}=
  Q_{2i+1,2i}= 1$, $Q_{2i,2i+1}=2-\epsilon$ and
  $Q_{2i+1,2i+1}=0$ for $i=0,1,2,3,4$.
In boxes we have the assignment of colors determined by 
algorithm {\tt Balance} (left hand side) and by algorithm {\tt B'}:
the top entry (bottom entry) contains the weight 
correspondent to the assigned color (discarded color). We have $Cost_{\tt B} \geq
(3-\epsilon) Cost_{\tt B'}$. {\tt CONTROLLARE LA FIGURA!!}
}
\end{figure}
}

Even if the approximability result is tight, if we are willing to pay
something in message complexity, we can get a $2$-approximation
algorithm. 

\begin{corollary}
Algorithm  {\tt Balance} can be transformed
  into a $2$-approximation algorithm, by paying an additional
  multiplicative $O(\log p )$ factor in message complexity.
\end{corollary}

\begin{proof}
Algorithm {\tt Balance} is modified in the following way: colors in
stage $r$ of {\sc Step 2} in Phase 3 are assigned to the agent having
the largest number of items (falling in the interval $I_r$) and
not to the one close to the leader. This can be achieved by making the
agent forward on the ring, not only their choice of colors, but also
their $Q_{i,j}$s for those colors. This requires extra $O(\log p)$
bits per color, increasing total message complexity of such a
multiplicative factor.  

For what concerns the approximation factor, this modification to the
algorithm allows to restate the thesis of Lemma~\ref{stessostage}
without the $2$ multiplicative factor and, following the same
reasoning of Theorem~\ref{3approx}, conclude the proof. 

\end{proof}

Finally, if we are not willing to pay extra message complexity, but we are
allowed to  wait for a longer time, we get a $(2+ \epsilon)$-approximation algorithm.

\begin{theorem} 
Assuming $m \in O(n^c)$, for some constant $c$, 
for any $0< \epsilon < 1$, there is a $(2+ \epsilon)$-approximation
algorithm for the Distributed Balanced  Color Assignment Problem with
running time $O(n \log_{1 + \epsilon} p) $ and message complexity
$\Omega (nm)$.
\end{theorem}

\begin{proof} 
Modify the two interval threshold values of algorithm {\tt Sync-Balance} in the
following way: 
$$l_r= \left\{ {\frac {{\hat p}}{(1 + \epsilon)^{r+1}}}
\right\} \mbox{ and }  u_r = \left\{ {\frac {{\hat p}}{(1 + \epsilon)^r}}
\right\}, $$
and redefine 

\begin{equation}
\label{eq:interval2}
 {\left\{ {\frac {a}{b}}
\right\} } = \left\{ \begin{array}{ll}
                        {\lceil\frac {a}{b}\rceil} & \mbox{ if } {\frac
{a}{b}}>
                        {\frac {1}{1+\epsilon}}; \\
                        0 & \mbox{ otherwise.}
                        \end{array} \right.
\end{equation}

Accordingly, the statement of Lemma~\ref{stessostage}
becomes $Q_{j,i} \leq (1+\epsilon) Q_{j,k}$, and the statement of
Lemma~\ref{massimo} can be rewritten as 

$$ Q_{j,\piO (j)} \leq \max \{ (1+\epsilon)\cdot Q_{j,\piB (j)} ,
Q_{k,\piB (k)} \}.$$

The result on the approximation factor then
follows by the same arguments of the proof of
Theorem~\ref{3approx}. The message complexity is not affected by these
changes, while the running time now depends on the number of stages in
Phase 3, that is $O(\log_{1 + \epsilon} p).$
\end{proof}

\section{Conclusion}

In this paper we have considered the Distributed Balanced Color
Assignment problem, which we showed to be the distributed version of 
different matching problems. In the distributed setting, the problem 
models situations where agents search a common space and need to
rearrange or organize the retrieved data. 

Our results indicate that these kinds of problems can be solved quite
efficiently { on a ring}, and that the loss incurred by
the lack of centralized control is not significant. We have focused
our attention to distributed solutions tailored for a ring of
agents. A natural extension would be to consider different topologies
and analyze how our techniques 
and ideas have to be modified in order to give efficient algorithms in more general
settings. 
We believe that the main ideas contained in this work could be useful
to extend the results even to arbitrary topologies. Indeed, 
an $O(n$ polylog$(n))$ distributed leader election protocol (that is needed in our algorithm)
is also available for arbitrary ad hoc radio networks \cite{CKP}.

For what concerns the ring topology, it is very interesting to note
that the value  $p$ never appears in the message complexity
for the synchronous case (not even if the polynomial relation between
$m$ and $n$ does not hold), while a factor $\log p$ appears in the
asynchronous case. It is still an open question if it is  
possible to devise an asynchronous  algorithm that shows optimal message complexity,
under the same hypothesis of the synchronous one; {\em i.e.}, if it is
possible to eliminate the extra $\log p / \log n$ factor. 

\paragraph{Acknowledgments.} The authors wish to thank Bruno Codenotti
for many helpful comments and discussions.

\end{document}